\documentclass[10pt]{amsart}
\usepackage[english]{babel}
\usepackage{amssymb}
\usepackage{amsmath}
\usepackage{srcltx}
\usepackage{lscape}

\newtheorem{dummy}{Dummy}

\newtheorem{theorem}[dummy]{Theorem}
\newtheorem{proposition}[dummy]{Proposition}
\newtheorem{corollary}[dummy]{Corollary}

\theoremstyle{definition}
\newtheorem{definition}{Definition}

\newtheorem{remark}[dummy]{Remark}

\newcommand{\ignore}[1]{}

\date{23.3.2015}
\author{S. Pumpl\"un}
\email{susanne.pumpluen@nottingham.ac.uk}
\address{School of Mathematical Sciences\\
University of Nottingham\\
University Park\\
Nottingham NG7 2RD\\
United Kingdom
}


\begin{document}

\title[Linear codes and nonassociative algebras  from skew-polynomial rings]
{A note on linear codes and nonassociative algebras obtained from skew-polynomial rings}

\begin{abstract}
Different approaches to construct linear codes using skew polynomials can be unified by using the nonassociative algebras built from
skew-polynomial rings by Petit.
\end{abstract}

\maketitle

\section*{Introduction}

In recent years, several classes of linear codes were obtained  from skew-polynomial rings (also called Ore rings).
Using this approach, self-dual codes with a better minimal
distance for certain lengths than previously known  were constructed:
while the classical
 \emph{cyclic codes} of length $m$ over a finite field $\mathbb{F}_q$ are obtained from ideals
 in the commutative ring $\mathbb{F}_q[t]/(t^m-1)$, and
\emph{constacyclic codes}
from ideals in  the commutative ring $\mathbb{F}_q[t]/(t^m-d)$, $d\in\mathbb{F}_q$,
\emph{ideal $\sigma$-codes} are associated with left ideals
in the non-commutative ring
$\mathbb{F}_q[t;\sigma]/(t^m-1)$, where $t^m-1\in R$ is a two-sided element in the twisted polynomial ring
 $\mathbb{F}_q[t;\sigma]$,
 $\sigma\in {\rm Aut}(\mathbb{F}_q)$,
see \cite{BU09.2}. Because $t^m-1$ is required to be a two-sided element in order for
$\mathbb{F}_q[t;\sigma]/(t^m-1)$ to be a ring,
 this enforces restrictions on the possible lengths
of the codes obtained: $t^m-1$ is two-sided if and only if
the order $n$ of $\sigma$ divides $m$ \cite[(15)]{P66}.

If $Rf$ denotes the left ideal generated by an element
$f\in R$, $R$ a ring, then $R/Rf$ is a left $R$-module.
 In \cite{BU09}, linear codes associated with left $R$-submodules
$Rg/Rf$  of  $R/Rf$ are considered, where $R=\mathbb{F}_q[t;\sigma]$ and $g$ is a right divisor of $f$. These codes are
called \emph{module $\sigma$-codes}.
 Another generalization is discussed in \cite{BU14} and \cite{BL13}, where codes obtained from
 submodules of the $R$-module $R/Rf$
 for some monic  $f\in R$ are investigated, where now $R=\mathbb{F}_q[t;\sigma,\delta]$ is a skew-polynomial ring.

We show that all these approaches can be unified since the codes mentioned above are associated to the left ideals
of the nonassociative algebra $S_f$ defined by Petit \cite{P66}. For a unital division ring $D$
(which here will be a finite field), and a polynomial
$f$ in the skew-polynomial ring $R=D[t;\sigma,\delta]$, Petit defined a nonassociative ring on the set
$R_m=\{h\in D[t;\sigma,\delta]\,|\, {\rm deg}(h)<m \}$, using right division
$g\circ h=gh \,\,{\rm mod}_r f $ to define the algebra multiplication.
$S_f=(R_m,\circ)$ is a nonassociative algebra over $F_0=\{a\in D\,|\, ah=ha \text{ for all } h\in S_f\}$
whose left ideals are generated by the polynomials $g$ which are right divisors of $f$.

The scenarios treated with respect to the linear codes mentioned above all require $f$ to be reducible, so the corresponding, not necessarily
associative, algebra $S_f$ is not allowed to be a division algebra here.
 The cyclic submodules studied in \cite{BU14}, \cite{BU14.2}  are exactly the left ideals in the algebra $S_f$.
The $(\sigma,\delta)$-codes of \cite{BL13} are the codes $\mathcal{C}$ associated to a left ideal of $S_f$ generated by a right divisor $g$ of $f$
 with $f\in K[t;\sigma,\delta]$. We show that if
  $\sigma$ is an automorphism of $K=\mathbb{F}_q$ and $\mathcal{C}$ a linear code over $\mathbb{F}_q$ of length $m$,
then $\mathcal{C}$ is a $\sigma$-constacyclic code with constant $d$ iff the skew-polynomial representation
$\mathcal{C}(t)$ with elements $a(t)$ obtained from $(a_0,\dots,a_{m-1})\in \mathcal{C}$
is a left ideal of $S_f$ with $f=t^m-d\in R=\mathbb{F}_q[t;\sigma]$,  generated by a monic right divisor $g$ of
$f$ in $R$.

\section{Preliminaries}

\subsection{Nonassociative algebras}

Let $F$ be a field and let $A$ be a finite-dimensional $F$-vector space.
We call $A$ an \emph{algebra} over $F$ if there exists an
$F$-bilinear map $A\times A\to A$, $(x,y) \mapsto x \cdot y$, denoted simply by juxtaposition $xy$,
the  \emph{multiplication} of $A$.
An algebra $A$ is called \emph{unital} if there is
an element in $A$, denoted by 1, such that $1x=x1=x$ for all $x\in A$.
We will only consider unital algebras.

   An algebra $A\not=0$ is called a {\it division algebra} if for any $a\in A$, $a\not=0$,
the left multiplication  with $a$, $L_a(x)=ax$,  and the right multiplication with $a$, $R_a(x)=xa$, are bijective.
$A$ is a division algebra if and only if $A$ has no zero divisors (\cite{Sch}, pp. 15, 16).

\subsection{Skew-polynomial rings}

In the following, we use results by Jacobson \cite{J96} and Petit \cite{P66}.
Let $D$ be a unital associative division ring, $\sigma$ a ring endomorphism of $D$ and $\delta$ a \emph{left $\sigma$-derivation} of $D$, i.e.
an additive map such that
$$\delta(ab)=\sigma(a)\delta(b)+\delta(a)b$$
for all $a,b\in D$, implying $\delta(1)=0$. The \emph{skew-polynomial ring} $D[t;\sigma,\delta]$
is the set of polynomials
$$a_0+a_1t+\dots +a_nt^n$$
with $a_i\in D$, where addition is defined term-wise and multiplication by
$$ta=\sigma(a)t+\delta(a) \quad (a\in D).$$
 $D[t;\sigma]=D[t;\sigma,0]$ is called a \emph{twisted polynomial ring} and
$D[t;\delta]=D[t;id,\delta]$ a \emph{differential polynomial ring}.
For the special case that $\sigma=id$ and $\delta=0$, we obtain the usual ring of left polynomials $D[t]=D[t;id,0]$, often also
 denoted
 $D_L[t]$ in the literature, with its multiplication given by
 $$(\sum_{i=1}^sa_it^i)(\sum_{i=1}^tb_it^i)=\sum_{i,j}a_ib_jt^{i+j}.$$
  If $D$ has finite dimension over its center and $\sigma$ is a ring automorphism of $D$, then
 $R=D[t;\sigma,\delta]$ is either a twisted polynomial or a differential polynomial ring by a linear change of variables \cite[Thm. 1.2.21]{J96}.
Note also that if $\sigma$  and $\delta$ are $F$-linear maps then $D[t;\sigma,\delta]\cong D[t]$ by a linear change
of variables.

For $f=a_0+a_1t+\dots +a_nt^n$ with $a_n\not=0$ define ${\rm deg}(f)=n$ and ${\rm deg}(0)=-\infty$.
Then ${\rm deg}(fg)={\rm deg} (f)+{\rm deg}(g).$
 An element $f\in R$ is \emph{irreducible} in $R$ if it is no unit and  it has no proper factors, i.e if there do not exist $g,h\in R$ with
 ${\rm deg}(g),{\rm deg} (h)<{\rm deg}(f)$ such
 that $f=gh$.

 $R=D[t;\sigma,\delta]$ is a left principal ideal domain (i.e.,  every left ideal in $R$ is of the form $Rf$)  and
there is a right-division algorithm in $R$ \cite[p.~3]{J96}: for all $g,f\in R$, $g\not=0$, there exist unique $r,q\in R$,
and ${\rm deg}(r)<{\rm deg}(f)$, such that
$$g=qf+r.$$
(We point out that our terminology is the one used by Petit \cite{P66}, Lavrauw and Sheekey \cite{LS}, and in the coding
literature;
it is different from Jacobson's \cite{J96}, who calls what we call right a left division algorithm and vice versa.)

\subsection{How to obtain nonassociative division algebras from skew-polynomial rings}

Let $D$ be a unital associative division algebra, $\sigma$ an injective ring homomorphism and
 $f\in D[t;\sigma,\delta]$ of degree $m$.

 \begin{definition} (cf. \cite[(7)]{P66})
  Let ${\rm mod}_r f$ denote the remainder of right division by $f$.
 Then  $$R_m=\{g\in D[t;\sigma,\delta]\,|\, {\rm deg}(g)<m\}$$
  together with the multiplication
 $$g\circ h=gh \,\,{\rm mod}_r f $$
 becomes a unital nonassociative algebra $S_f=(R_m,\circ)$ over
 $$F_0=\{a\in D\,|\, ah=ha \text{ for all } h\in S_f\}.$$
\end{definition}

 This algebra is also denoted by $R/Rf$ \cite{P66,P68} if we want to make clear which ring $R$ is involved in the
 construction.
 $F_0$  is a subfield of $D$ \cite[(7)]{P66}.

\begin{remark}
For $f(t)=t^m-d \in R=D[t;\sigma]$, the multiplication in $S_f$ is defined via
\[
 (at^i)(bt^j) =
  \begin{cases}
   a\sigma^i(b)   t^{i+j} & \text{if } i+j < m, \\
     a \sigma^i(b)  t^{(i+j)-m}d & \text{if } i+j \geq m,
  \end{cases}
\]
for all $a,b\in D$ and then linearly extended.

\end{remark}

\section{Linear codes associated to left ideals of $S_f$}

Let $K$ be a finite field, $\sigma$ an automorphism of $K$ and $F={\rm Fix}(\sigma)$, $[K:F]=n$.
By a linear base change we could always assume $\delta=0$. However, \cite{BU14}  and \cite{BL13} show that
 this  limits the choices of available codes.

 Unless specified otherwise, let $R=K[t;\sigma,\delta]$ and $f\in R$ be a monic polynomial of degree $m$.
Analogously as for instance
in  \cite{BU14}, \cite{BU14.2}, \cite{BU09}, \cite{BU09.2}, we associate to an element
$a(t)=\sum_{i=0}^{m-1}a_it^i$ in $S_f$ the  vector $(a_0,\dots,a_{m-1})$. Our codes
$\mathcal{C}$ of length $m$ consist of all $(a_0,\dots,a_{m-1})$ obtained this way from the elements
 $a(t)=\sum_{i=0}^{m-1}a_it^i$ in a left ideal $I$ of $S_f$.
 Conversely, for a linear code $\mathcal{C}$ of length $n$ we denote by $\mathcal{C}(t)$ the set of skew-polynomials
 $a(t)=\sum_{i=0}^{m-1}a_it^i\in S_f$ associated to the codewords $(a_0,\dots,a_n)\in \mathcal{C}$.

\begin{proposition}  Let $D$ be a unital associative division ring and $f\in R=D[t;\sigma,\delta]$.
\\ (i) All left ideals in $S_f$ are generated by some monic right divisor $g$ of $f$ in $R$.
\\ (ii) If $f$ is irreducible, then $S_f$ has no non-trivial left ideals.
\end{proposition}

\begin{proof}
(i) The proof is analogous to the one of \cite[Lemma 1]{BGU07}, only that now we are working in the  nonassociative
 ring $S_f$:
Let $I$ be a left ideal of $S_f$. If $I=\{0\}$ then $I=(0)$. So suppose $I\not=(0)$ and choose a monic non-zero
polynomial $g$ in $I\subset R_m$
of minimal degree. For $p\in I\subset R_m$, a right division by $g$ yields unique $r,q\in R$ with
${\rm deg}(r)<{\rm deg}(g)$ such that
$$p=qg+r $$
and hence $r=p-qg\in I$. Since we chose $g\in I$ to have minimal degree, we conclude that $r=0$, implying
$p=qg$ and so $I=Rg$.
\\ (ii) follows from (i).
\end{proof}

\begin{corollary}
 (i) The cyclic submodules studied in \cite{BU14}, \cite{BU14.2} are exactly the left ideals in the algebra $S_f$
where $f\in \mathbb{F}_q[t;\sigma,\delta]$.
\\ (ii) The $(\sigma,\delta)$-codes $\mathcal{C}$ in \cite{BL13} are exactly the codes  associated to a left ideal of $S_f$ generated by a
non-trivial right divisor $g$ of $f\in D[t;\sigma,\delta]$,
 whenever $D$ is an associative division ring (usually, $D=\mathbb{F}_q$).
\end{corollary}

  Note that when we look at the nonassociative  case, where $f$ is not two-sided anymore,
  it can happen that $f$ is irreducible in $K[t;\sigma,\delta]$, hence does not have any non-trivial right divisors $g$.

\begin{remark}
Let $m\geq 2$.
Since for $a(t)\in S_f$ also $ta(t)\in S_f$, we obtain  for  $f(t)=t^m-d\in K[t;\sigma]$ that
$$ta(t)=\sigma(a_0)t+\sigma(a_1)t^2+\dots+\sigma(a_{m-1})t^m=
\sigma(a_{m-1})d+\sigma(a_0)t+\sigma(a_1)t^2+\dots+\sigma(a_{m-2})t^{m-1}$$
in $S_f$,
so that
$$(a_0,\dots,a_{m-1})\in  \mathcal{C}\Rightarrow (\sigma(a_{m-1})d,\sigma(a_0),\dots,\sigma(a_{m-2}))\in  \mathcal{C}$$
 is a $\sigma$-constacyclic code (even if $S_f$ is division).
 With the same argument, every left ideal $Rg$ in $S_f$ with $g\in R$ a right divisor of $f=t^m-d$  yields a
$\sigma$-constacyclic code $\mathcal{C}$ for $d\not=1$ and a $\sigma$-cyclic code for $d=1$.
\end{remark}

In \cite[Theorem 1]{BGU07} it is shown that the code words of a $\sigma$-cyclic code are coefficient tuples of elements
$a(t)=\sum_{i=0}^{m-1}a_it^i\in \mathbb{F}_q[t;\sigma]/(t^m-1)$, which are left multiples of some element $g\in
\mathbb{F}_q[t;\sigma]/(t^m-1)$ which is a right divisor of $f$, under the assumption that
 the order $n$ of $\sigma$ divides $m$. The assumption that $n$ divides $m$ guarantees that $Rf$ is a two-sided ideal,
  i.e.
 that $S_f$ is associative, but is not required:

\begin{theorem}
Let $\sigma$ be an automorphism of $K=\mathbb{F}_q$ and $\mathcal{C}$ a linear code over $\mathbb{F}_q$ of length $m$.
Then $\mathcal{C}$ is a $\sigma$-constacyclic code (with constant $d$) iff the skew-polynomial representation
$\mathcal{C}(t)$ with elements $a(t)$ obtained from $(a_0,\dots,a_{m-1})\in \mathcal{C}$
is a left ideal of $S_f$ with $f=t^m-d\in R=\mathbb{F}_q[t;\sigma]$,  generated by a monic right divisor $g$ of
$f$ in $R$.
\end{theorem}

\begin{proof}
$\Leftarrow:$ This is Remark 8.

$\Rightarrow:$ The argument is analogous to the proof of \cite[Theorem 1]{BGU07}:
 If we have a $\sigma$-constacyclic code $\mathcal{C}$, then its elements
  define polynomials $a(t)\in \mathbb{F}_q[t;\sigma]=K[t;\sigma]$. These form a left ideal
  $\mathcal{C}(t)$ of $S_f$
  with $f=t^m-d\in \mathbb{F}_q[t;\sigma]$:
  The code is linear  and so the skew-polynomial
  representation $\mathcal{C}(t)$ is an additive group. For $(a_0,\dots,a_{m-1})\in \mathcal{C}$,
  $$ta(t)=\sigma(a_0)t+\sigma(a_1)t^2+\dots+\sigma(a_{m-1})t^m$$
and since $f=t^m-d$ we get in $S_f=R/Rf$ that
  $$ta(t)=\sigma(a_{m-1})d+\sigma(a_0)t+\sigma(a_1)t^2+\dots+\sigma(a_{m-2})t^{m-1}.$$
  Since $\mathcal{C}$ is $\sigma$-constacyclic with constant $d$, $ta(t)\in \mathcal{C}(t)$.
  Clearly, by iterating this argument, also $t^s a(t)\in \mathcal{C}(t)$ for all $s\leq m-1$.
By iteration and linearity of $\mathcal{C}$, thus $h(t)a(t)\in \mathcal{C}(t)$ for all $h(t)\in R_m$, so
$\mathcal{C}(t)$ is closed under multiplication
and a left ideal of $S_f$.
\end{proof}

\begin{corollary}
Let $\sigma$ be an automorphism of $K=\mathbb{F}_q$ and $\mathcal{C}$ a linear code over $\mathbb{F}_q$ of length $m$.
Then $\mathcal{C}$ is a $\sigma$-cyclic code iff the skew-polynomial representation
$\mathcal{C}(t)$ with elements $a(t)$ obtained from $(a_0,\dots,a_{m-1})\in \mathcal{C}$
is a left ideal  of $S_f$ generated by a monic right divisor $g$ of $f=t^m-1\in R=\mathbb{F}_q[t;\sigma]$.
\end{corollary}

\begin{remark}\label{re:8}
Let  $f(t)=t^m-d\in R=K[t;\sigma]$ and $F={\rm Fix}(\sigma)$. Then $f$ is a two-sided element
(thus $S_f$ associative and $f$ reducible)
  iff $m$ divides the order $n$ of $\sigma$
  and $d\in F$  by \cite[(7), (9)]{P66}. For $d=1$ in particular, $f$ is two-sided iff
$m$ divides the order $n$ of $\sigma$.
Any right divisor $g$  of degree $k$ of, for instance, $f=t^m-d$
 can be used to construct a $\sigma$-constacyclic $[m,m-k]$-code (with constant $d$).
  (If $f$ is not two-sided,
  it can happen that $f$ is irreducible in $K[t;\sigma]$, hence
 does not have any non-trivial right divisors $g$.)
 We note:
\\ (i) $f(t)=t^3-d$ is reducible in $R$ if and only if
$$\sigma(z)^2\sigma(z)z=d \text{ or }\sigma(z)^2\sigma(z)z=d$$
 for some $z\in K$ \cite[(18)]{P66} . (Thus $t^3-1$ is always reducible in $K[t;\sigma].$)
 \\ (ii) Suppose $m$ is prime and $F$ contains a primitive $m$th root of unity.
Then $f(t)=t^m-d$ is reducible in $R$  if and only if
$$d=\sigma^{m-1}(z)\cdots\sigma(z)z
\text{ or }\sigma^{m-1}(d)=\sigma^{m-1}(z)\cdots\sigma(z)z$$
 for some $z\in K$ \cite[(19)]{P66}. (Thus $t^m-1$ is always reducible in $K[t;\sigma]$, if
 $F$ contains a primitive $m$th root of unity.)
\\ (iii) Let $K/F$  have degree $m$, ${\rm Gal}(K/F)=\langle\sigma\rangle$ and $R=K[t;\sigma]$,
$f=t^m-d$ with $d\not\in F$. Then the nonassociative cyclic algebra $(K/F,\sigma,d)$ is the algebra
$S_f$ with $R=K[t;\sigma^{-1}]$ and $f(t)=t^m-d$ (cf. \cite[p.~13-13]{P66}).
\\ (a) If the elements $1,d,\dots,d^m$ are linearly dependent over $F$, then $f$ is reducible.
\\ (b) If $m$ is prime  then $f$ is irreducible \cite{S12} and thus there are no $\sigma$-constacyclic codes with constant
$d$ apart from the $[m,m]$-code associated with $S_f$ itself.
\end{remark}

When working over finite fields, the division algebras $S_f$ are finite semifields which are closely related to the semifields constructed by
Johnson and Jha \cite{JJ} obtained by employing semi-linear transformations. Results for these semifields and their spreads might be
useful for future linear code constructions.

\end{document}